\documentclass[12pt,english,ngerman,british]{amsart}
\usepackage[foot]{amsaddr}
\usepackage[T1]{fontenc}
\usepackage[latin9]{inputenc}
\usepackage{geometry}
\usepackage{babel}
\usepackage{amsthm}
\usepackage{amssymb}
\usepackage{graphicx}
\usepackage[unicode=true,pdfusetitle,
 bookmarks=true,bookmarksnumbered=false,bookmarksopen=false,
 breaklinks=false,pdfborder={0 0 1},backref=false,colorlinks=false]
 {hyperref}
\usepackage{breakurl}

\makeatletter

\providecommand{\tabularnewline}{\\}

\numberwithin{equation}{section}
\numberwithin{figure}{section}
\theoremstyle{plain}
\newtheorem{thm}{\protect\theoremname}[section]
  \theoremstyle{definition}
  \newtheorem{defn}[thm]{\protect\definitionname}
  \theoremstyle{plain}
  \newtheorem{cor}[thm]{\protect\corollaryname}
  \theoremstyle{plain}
  \newtheorem{prop}[thm]{\protect\propositionname}
  \theoremstyle{remark}
  \newtheorem{claim}[thm]{\protect\claimname}

\usepackage{listings}
\usepackage{color}

\makeatother

  \addto\captionsbritish{\renewcommand{\claimname}{Claim}}
  \addto\captionsbritish{\renewcommand{\corollaryname}{Corollary}}
  \addto\captionsbritish{\renewcommand{\definitionname}{Definition}}
  \addto\captionsbritish{\renewcommand{\propositionname}{Proposition}}
  \addto\captionsbritish{\renewcommand{\theoremname}{Theorem}}
  \addto\captionsenglish{\renewcommand{\claimname}{Claim}}
  \addto\captionsenglish{\renewcommand{\corollaryname}{Corollary}}
  \addto\captionsenglish{\renewcommand{\definitionname}{Definition}}
  \addto\captionsenglish{\renewcommand{\propositionname}{Proposition}}
  \addto\captionsenglish{\renewcommand{\theoremname}{Theorem}}
  \addto\captionsngerman{\renewcommand{\claimname}{Behauptung}}
  \addto\captionsngerman{\renewcommand{\corollaryname}{Korollar}}
  \addto\captionsngerman{\renewcommand{\definitionname}{Definition}}
  \addto\captionsngerman{\renewcommand{\propositionname}{Satz}}
  \addto\captionsngerman{\renewcommand{\theoremname}{Theorem}}
  \providecommand{\claimname}{Claim}
  \providecommand{\corollaryname}{Corollary}
  \providecommand{\definitionname}{Definition}
  \providecommand{\propositionname}{Proposition}
\providecommand{\theoremname}{Theorem}

\begin{document}

\title{Expected Coverage of Random Walk Mobility Algorithm}

\author{Moritz Kohls} 
\address{Department of Electrical Engineering and Information Technology\newline
	Technische Universität Dortmund, Germany}
\email{moritz.kohls@udo.edu}

\author{Tanja Hernandez}
\email{tanja.hernandez@hs-owl.de}

\begin{abstract}
Unmanned aerial vehicles (UAVs) have been increasingly used for exploring
areas. Many mobility algorithms were designed to achieve a fast coverage
of a given area. We focus on analysing the expected coverage of the
symmetric random walk algorithm with independent mobility. Therefore
we proof the dependence of certain events and develop Markov models,
in order to provide an analytical solution for the expected coverage.
The analytic solution is afterwards compared to those of another work
and to simulation results.
\end{abstract}

\maketitle

\section{Introduction}

There are several reasons to use UAVs instead of helicopters for exploring
the area of interest. The most important advantage should be the safety.
The life of the crew is risked, when the helicopter comes into operation
in dangerous places. In some difficult surroundings there might be
too many obstacles nearby, where the helicopter cannot provide sufficient
mobility to discover the area. Even in flat terrain it argues for
the commitment of drones, because a helicopter flight is associated
with great costs. In case of emergency it is necessary to discover
the place of interest as fast as possible. A network of UAVs can work
together to achieve this goal faster than a single helicopter could.

One of the simplest mobility algorithms for this purpose is the bordered
symmetric random walk. However, it may be supposed that
this algorithm provides a rather weak coverage performance.
The aim of this work is therefore to deduce an analytical solution for the expected
coverage performance of this algorithm, so that one can compare the coverage
performance to those of other mobility algorithms. In especially,
our results of the expected coverage of the bordered symmetric random
walk are compared to those of another paper \cite{key-1}, to check the validation
of the analytical results.

The composition of this work is made as follows. In chapter 2 related
work is presented, in which another analytic solution is already contained.
Chapter 3 involves the analysis of the random walk mobility algorithms.
At first, bordered and boundless Markov chains for the mobility algorithms
are proposed. Subsequently, the coverage performance using the example
of expected coverage is calculated. Finally, the two analytical solutions
are compared to the simulation results.

\section{Related work}

An analytical approach to the calculation of the expected coverage
is given in \cite{key-1}. A stochastic process is defined, which
enables the authors to provide a formula for coverage metrics like
expected coverage and full coverage probability. However, the formulas
of the metrics derived from the examined Markov model base upon not
mentioned preconditions, which do not apply to Markov models generally.
By taking the preconditions of the Markov model into consideration,
the analytical solution for the expected coverage becomes considerably
more complicated. The exposure to Markov chains is executed in \cite{key-2}.
The numeric programming is accomplished by statistical software R
\cite{key-3}.

\section{Analysis of the Markov chains}

\subsection{Markov Chain}

At first, it is necessary to introduce the mobility algorithm that
is investigated in this paper. Foremost, we focus on examining the
mobility algorithm in a two-dimensional exploration area. Indeed, the
three-dimensional case provides similar results, but with the great
drawback that with an increasing number of states the
computation time for realistic large areas increases unnecessarily. Another
practical problem of the three-dimensional model would be the huge
size of the required matrices, which have to be multiplied for the analytic
solution. This often exceeds the limits of the utilised numerical
software. Thus, for now we constrain the exploration area of the UAVs
to a two-dimensional lattice. For convenience, we determine the position
of the drones as a finite set of states where they can be located
on discrete time points. The units of location and time can be chosen
arbitrarily.

The UAVs move randomly, therefore we deal with a discrete-value and
discrete-time stochastic process, which equates to a homogeneous and
discrete-time Markov chain $M:=(E,\pi_{0},P)$. The state space $E$
contains every point on a two-dimensional grid with width $w$ and
depth $d$ for the two-dimensional case and additionally height $h$
for the three-dimensional case. One could define these points as two-dimensional
vectors $(i,j)$ for $i=1,\ldots,w$ and $j=1,\ldots,d$, however,
for the mathematical handling we need to sequence the individual states
in order to receive a one-dimensional numbering. For this purpose
we sequence the states on the considered rectangular grid as in table
\ref{tab:Sequence-of-all}. An algorithm for sequencing an arbitrary
three-dimensional grid is provided in the appendix.
\begin{table}
\begin{tabular}{|c|c|c|c|c|c|}
\hline 
$1$ & $2$ & $\cdot\cdot\cdot\cdot\cdot\cdot\cdot\cdot\cdot\cdot$ & $\cdot\cdot\cdot\cdot\cdot\cdot\cdot\cdot\cdot\cdot$ & $w-1$ & $w$\tabularnewline
\hline 
$w+1$ & $w+2$ & $\cdot\cdot\cdot\cdot\cdot\cdot\cdot\cdot\cdot\cdot$ & $\cdot\cdot\cdot\cdot\cdot\cdot\cdot\cdot\cdot\cdot$ & $2w-1$ & $2w$\tabularnewline
\hline 
$\cdot\cdot\cdot\cdot\cdot\cdot\cdot\cdot\cdot\cdot$ & $\cdot\cdot\cdot\cdot\cdot\cdot\cdot\cdot\cdot\cdot$ & $\cdot\cdot\cdot\cdot\cdot\cdot\cdot\cdot\cdot\cdot$ & $\cdot\cdot\cdot\cdot\cdot\cdot\cdot\cdot\cdot\cdot$ & $\cdot\cdot\cdot\cdot\cdot\cdot\cdot\cdot\cdot\cdot$ & $\cdot\cdot\cdot\cdot\cdot\cdot\cdot\cdot\cdot\cdot$\tabularnewline
\hline 
$\cdot\cdot\cdot\cdot\cdot\cdot\cdot\cdot\cdot\cdot$ & $\cdot\cdot\cdot\cdot\cdot\cdot\cdot\cdot\cdot\cdot$ & $\cdot\cdot\cdot\cdot\cdot\cdot\cdot\cdot\cdot\cdot$ & $\cdot\cdot\cdot\cdot\cdot\cdot\cdot\cdot\cdot\cdot$ & $\cdot\cdot\cdot\cdot\cdot\cdot\cdot\cdot\cdot\cdot$ & $\cdot\cdot\cdot\cdot\cdot\cdot\cdot\cdot\cdot\cdot$\tabularnewline
\hline 
$(d-2)w+1$ & $(d-2)w+2$ & $\cdot\cdot\cdot\cdot\cdot\cdot\cdot\cdot\cdot\cdot$ & $\cdot\cdot\cdot\cdot\cdot\cdot\cdot\cdot\cdot\cdot$ & $(d-1)w-1$ & $(d-1)w$\tabularnewline
\hline 
$(d-1)w+1$ & $(d-1)w+2$ & $\cdot\cdot\cdot\cdot\cdot\cdot\cdot\cdot\cdot\cdot$ & $\cdot\cdot\cdot\cdot\cdot\cdot\cdot\cdot\cdot\cdot$ & $dw-1$ & $dw$\tabularnewline
\hline 
\end{tabular}

\medskip{}

\caption{Sequence of all $dw$ states of the two-dimensional exploration area
with width $w$ and depth $d$. \label{tab:Sequence-of-all}}
\end{table}

So, the state space consists of $d\cdot w$ single states $E:=\{1,\ldots,dw\}$.
The starting distribution $\pi_{0}$ contains the probabilities for
the points, in which the drones start their flight at point in time $0$. If all UAVs start from the same point $i\in\{1,\ldots,dw\}$,
we would choose $\pi_{0}$ as the unit vector $e_{i}$. If one refuses
to decide on a deterministic starting point, $\pi_{0}:=\left(\frac{1}{dw},\ldots,\frac{1}{dw}\right)$
is a reasonable election for the starting distribution, because then
every state has an equal probability for the start.

The examined mobility algorithm determines the transition probabilities
$p_{ij}$ for a transition from state $i$ to state $j$. The transition
probabilities form the transition matrix $P$, which has as many rows
and columns as the Markov chain has different states, thus $dw$ for
the two-dimensional example. The distance between two places can be
calculated by using the taxicab geometry. So the drones are only allowed
to fly into at most four different directions in a two-dimensional
world and six directions (dimensions times two) in the three-dimensional
world.

\selectlanguage{british}%
\begin{figure}
\noindent \begin{centering}
\includegraphics[width=9cm,height=6cm]{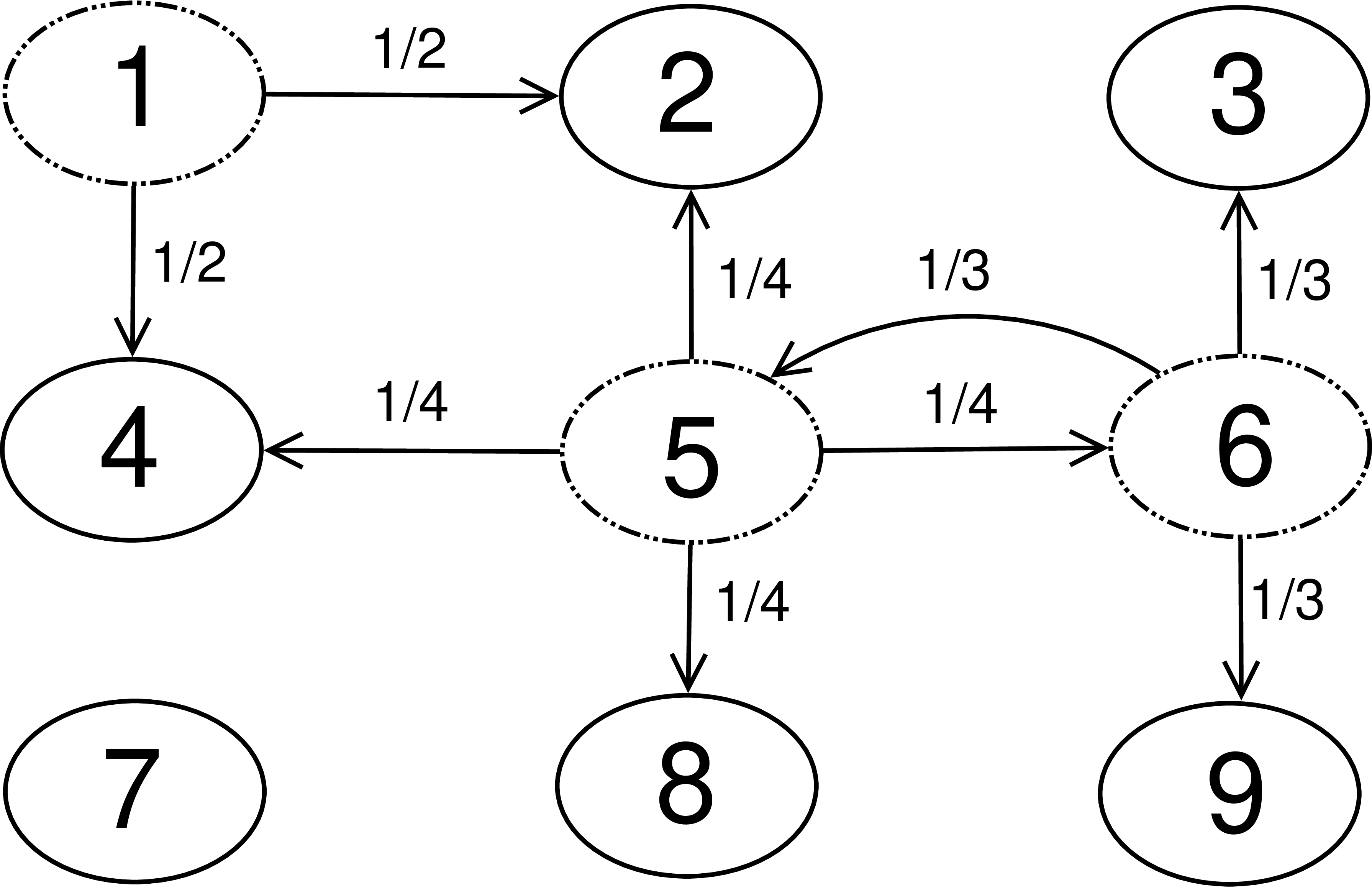}
\par\end{centering}

\selectlanguage{british}%
\medskip{}

\selectlanguage{british}%
\caption{\selectlanguage{british}%
Transition probabilities of the bordered symmetric random walk of
the states $1$, $5$ and $6$. \label{fig:Transition-probabilities-of}\selectlanguage{ngerman}%
}
\end{figure}

\begin{figure}
\includegraphics[width=9.75cm,height=6.5cm]{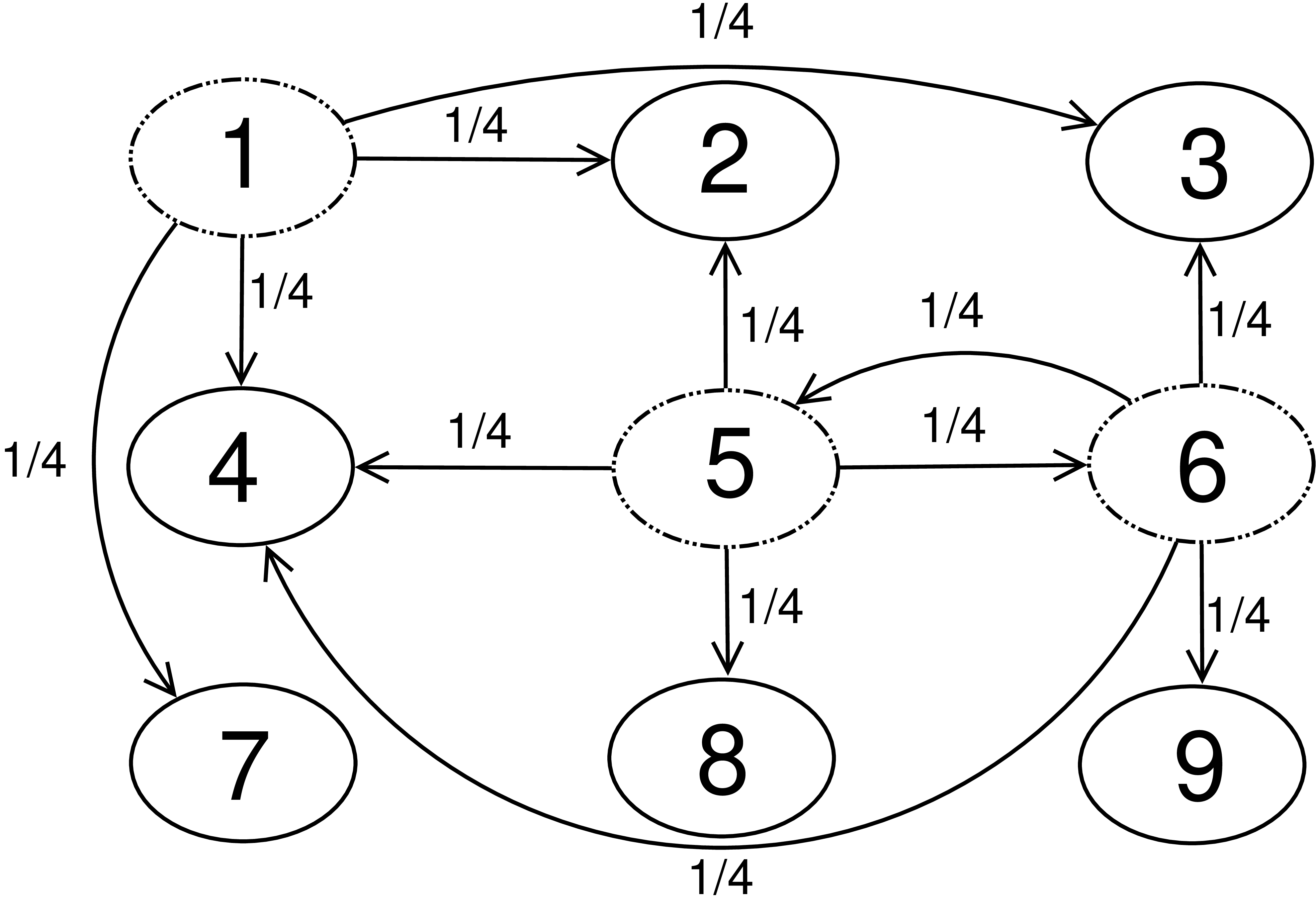}

\selectlanguage{british}%
\medskip{}

\selectlanguage{british}%
\caption{\selectlanguage{british}%
Transition probabilities of the symmetric random walk without borders
of the states $1$, $5$ and $6$. \label{fig:Transition-probabilities-of-1}\selectlanguage{ngerman}%
}
\end{figure}

\selectlanguage{british}%
Figure \ref{fig:Transition-probabilities-of} shows the transition
probabilities for the firstly named application. Inside of the exploration
area there is an equal probability of $\frac{1}{4}$ for flying forward,
backward, to the left or right on a symmetric random walk. On the
brinks of the area three directions remain and in the corners only
two flight directions are still possible, at least when borders exist
which should be standard case. Without borders the drones can leave
the grid, only to reappear on the opposite side (see figure \ref{fig:Transition-probabilities-of-1}).
In this case the transition probabilities remain at $\frac{1}{4}$
even for the marginal and corner states. For example, we present the
transition matrix of the $3\times3$ bordered symmetric random walk
in table \ref{tab:Transition-matrix-of}, whereupon entries with probability
$0$ are omitted on grounds of clarity.

\begin{table}
$\begin{pmatrix} & \frac{1}{2} &  & \frac{1}{2}\\
\frac{1}{3} &  & \frac{1}{3} &  & \frac{1}{3}\\
 & \frac{1}{2} &  &  &  & \frac{1}{2}\\
\frac{1}{3} &  &  &  & \frac{1}{3} &  & \frac{1}{3}\\
 & \frac{1}{4} &  & \frac{1}{4} &  & \frac{1}{4} &  & \frac{1}{4}\\
 &  & \frac{1}{3} &  & \frac{1}{3} &  &  &  & \frac{1}{3}\\
 &  &  & \frac{1}{2} &  &  &  & \frac{1}{2}\\
 &  &  &  & \frac{1}{3} &  & \frac{1}{3} &  & \frac{1}{3}\\
 &  &  &  &  & \frac{1}{2} &  & \frac{1}{2}
\end{pmatrix}$

\medskip{}

\caption{Transition matrix of the bordered symmetric random walk with width
and depth 3. \label{tab:Transition-matrix-of}}
\end{table}

\subsection{Calculation of the Expected Coverage}

It takes some definitions concerning random variables to calculate the expected
coverage of the symmetric random walk $\left(X_{n}\right)_{n\in\mathbb{N}_{0}}$.
Let $C_{n,z}$ indicate, if drone $d$ has visited state $z$ until
point in time $n$, thus whether a coverage occurred ($C_{n,z}=1$)
or not. So, it is essential that
\begin{equation}
C_{n,z}:=\begin{cases}
1 & if\,X_{0}=z\vee X_{1}=z\vee...\vee X_{n}=z\\
0 & else
\end{cases}.
\end{equation}
These random variables are Bernoulli-distributed with for now unknown
parameters $p_{n,z}$. It should be mentioned that for a constant
point in time $n$ the random variables $C_{n,z}$ for $z=1,\ldots,|E|$
neither are stochastically independent nor identical distributed which
can be proven trivially. The same proposition can be made for all
$n\in\mathbb{N}_{0}$, if a particular state $z$ is considered. We
will now confirm the dependence of the random variables in the last
named case. But primarily, some characteristics of probability theory
are needed.
\begin{defn}
Independence of Events

Let $A_{1},\ldots,A_{n}$ be events at the same probability space.
The events $A_{1},\ldots,A_{n}$ are stochastically independent if
and only if for all $k\in\{1,...,n\}$ and every index space $I=\{i_{1},...i_{k}\}\subseteq\{1,...,n\}$
of the cardinality $k$ 
\begin{equation}
P(A_{i_{1}}\cap...\cap A_{i_{k}})=P(A_{i_{1}})\cdot...\cdot P(A_{i_{k}})
\end{equation}
is valid.\end{defn}
\begin{cor}
\label{cor:Independence-of-Events}Independence of Events

The events $A$ and $B$ are stochastically independent, if and only
if $P(B)=0$ or the equation 
\begin{equation}
P(A|B)=P(A)
\end{equation}
 is satisfied for $P(B)>0$.\end{cor}
\begin{prop}
State Probabilities of a Markov Chain

The starting distribution $\pi_{0}$ is a $|E|$- dimensional row
vector and contains the probabilities for starting in each of the
states $z\in E$. The state probabilities after $n$ time steps are
calculated by
\[
\pi_{n}=\pi_{0}\cdot P^{n}.
\]

\end{prop}
Below we proof the mentioned claim, which is the central conclusion
of this paper.
\begin{prop}
\label{prop:Dependence-of-Coverage}Dependence of Coverage Events
(1)
\end{prop}
The events $\{X_{0}=z\},...,\{X_{n}=z\}$ are stochastically dependent
for $n\geq3$ and at least one state $z$.
\begin{proof}
The exploration area of the UAVs shall be big enough on grounds of
realistic scenarios which requires a minimum cell count of five in
each dimension. Apparently, the Markov chain holds $P(X_{t}=z,X_{t+1}=z)=0$
for every point in time and state, because the drones keep in motion.

\emph{Case 1:} The starting distribution has at least two states with
an uneven Manhattan-distance to each other, that both have a positive
probability on point in time zero.

In this case all states have a positive probability from a particular
point in time onwards. Mathematically expressed, $P(X_{t}=z)>0$ and $P(X_{t+1}=z)>0$
apply from a particular point in time $n$ 
 for every state $z$. 
In worst-case, $n$ constitutes still less than width + depth for the two-dimensional
case respectively width + depth + height for the three-dimensional
case. The worst-case occurs, if the UAVs can only start from a corner,
however state $z$ is located in the opposite corner. Not later than
at this particular time applies 
\[
P(X_{t}=z,X_{t+1}=z)=0<P(X_{t}=z)\cdot P(X_{t+1}=z),
\]
whereby the chronologically succeeding events $\{X_{t}=z\}$ and $\{X_{t+1}=z\}$
are stochastically dependent, likewise all chronologically succeeding
events for arbitrary $z\in E$.

A typical special case of the first case is the starting distribution
with an equal probability for every state. One realizes easily that
the requirement in first case is achieved, if a random walk without
borders has at least an uneven width, depth or height.

\emph{Case 2:} The starting distribution has no pair of states with
an uneven Manhattan-distance to each other.

A typical example is, if the drone takes off definitely from one particular
whereabouts. It needs at the most width + depth (+ height) time steps,
until approximately half of the states receive a positive probability.
(If $dw$ is even, there are exactly $\frac{dw}{2}$ such states for
the two-dimensional case; if $dw$ is uneven, the number of such states
switches between $\frac{dw-1}{2}$ and $\frac{dw+1}{2}$). Let this
point in time be notated as $m$. Now we pick a state $z$ and check
the independence property $P(X_{m+2}=z|X_{m}=z)=P(X_{m+2}=z)$ as
mentioned in corollary \ref{cor:Independence-of-Events}. The conditional
probability $P(X_{m+2}=z|X_{m}=z)$ is only defined, if $P(X_{m}=z)$
is greater than zero. Otherwise the events $\left\{ X_{m}=z\right\} $
and $\left\{ X_{m+2}=z\right\} $ would be stochastically independent,
which is true for up to half of the states in the bordered random
walk for the second case. Below we will show that both events are
even independent under the constraint $P(X_{m}=z)=1$.
\begin{eqnarray}
P(X_{m+2}=z|X_{m}=z) & = & P(X_{m+2}=z)\\
\Leftrightarrow P^{2}(z,z) & = & \sum_{i=1}^{|E|}P(X_{m}=i)P^{2}(i,z)\label{eq:3}\\
\Leftrightarrow P^{2}(z,z) & = & P(X_{m}=z)P^{2}(z,z)+\sum_{i=1,i\neq z}^{|E|}P(X_{m}=i)P^{2}(i,z)\\
 & \Leftarrow & P(X_{m}=z)=1
\end{eqnarray}
Hence it can be reasoned that both events are independent, if $P(X_{m}=z)\in\{0,1\}$.
This applies in the mentioned special case, if the drone departs definitely
from one particular state. Then $\left\{ X_{0}=z\right\} $, $\left\{ X_{1}=z\right\} $ and
$\left\{ X_{2}=z\right\} $ are independent for all states $z\in E$.
On the contrary, for different starting distributions, $\left\{ X_{0}=z\right\} $
and $\left\{ X_{2}=z\right\} $ can be stochastically dependent, because
at least two distinct states $z_{1}$ and $z_{2}$ provide the inequality
$0<\pi_{0}(z_{1}),\pi_{0}(z_{2})<1$.

To summarize, the events $\left\{ X_{m}=z\right\} $ and $\left\{ X_{m+2}=z\right\} $
are dependent if and only if 
\begin{equation}
P(X_{m}=z)\in\left(0,1\right)\label{eq:1}
\end{equation}
 and 
\begin{equation}
\left(1-P(X_{m}=z)\right)P^{2}(z,z)\neq\sum_{i=1,i\neq z}^{|E|}P(X_{m}=i)P^{2}(i,z).\label{eq:2}
\end{equation}
This should be valid in most examined Markov models at least for one
state $z\in E$ and $m\geq1$. Below, the second case is analysed
in more detail, in order to proof the dependence for some sub cases.
Note that for the missing sub cases the dependence characteristic
can be confirmed by checking the correctness of the above mentioned
formulas \ref{eq:1} and \ref{eq:2}, as soon as the concrete Markov
model is created.

\selectlanguage{english}%
\emph{Special Case of case 2:}\foreignlanguage{british}{ Bordered
Random Walk}

\selectlanguage{british}%
We will focus on the two-dimensional bordered random walk, firstly
with an uneven width and depth. The three-dimensional case proceeds
similarly. The starting distribution shall be symmetric for now, for
example the drone starts certainly in the middle of the grid. Let
the point of time $m$ be such that in the corners $1$, $w$, $(d-1)w+1$
and $dw$ the states have a positive probability. Notice that this
is only possible with an uneven width and depth. For the transition
probability in two steps one can easily deduce $P^{2}(1,1)=P^{2}(w,w)=P^{2}\left((d-1)w+1,(d-1)w+1\right)=P^{2}(dw,dw)=\frac{1}{3}$.
The unconditional probability $P(X_{m+2}=z)$ can be obtained by formula
\begin{equation}
P(X_{m+2}=z)=\sum_{i=1}^{|E|}P(X_{m}=i)P^{2}(i,z).
\end{equation}
This implicates for the top left corner the equality 
\begin{eqnarray*}
 &  & P(X_{m+2}=1)\\
 &  & =P(X_{m}=1)\cdot P^{2}(1,1)+P(X_{m}=3)\cdot P^{2}(3,1)\\
 &  & +P(X_{m}=2w+1)\cdot P^{2}(2w+1,1)+P(X_{m}=w+2)\cdot P^{2}(w+2,1)\\
 &  & =\frac{1}{3}P\left(X_{m}=1\right)+\frac{1}{9}P\left(X_{m}=3\right)+\frac{1}{9}P\left(X_{m}=2w+1\right)+\frac{1}{12}P\left(X_{m}=w+2\right)\\
 &  & =\frac{1}{3}
\end{eqnarray*}
 as a necessary and sufficient constraint for the independence. Due
to the symmetric starting distribution all corner states provide the
same probability $P(X_{m+2}=1)=P(X_{m+2}=w)=P\left(X_{m+2}=(d-1)w+1)\right)=P(X_{m+2}=dw)$
whose sum can be at most $1$. The outcome of this is $P(X_{m+2}=1)\leq\frac{1}{4}$
which is less than $\frac{1}{3}$, hence the condition \ref{eq:3}
is not pervaded and therefore the events are dependent.

In some settings the starting distribution cannot be modelled as symmetric
and the recent proof cannot be executed in an analogous manner, although
the probabilities for the corner states will converge to the same
positive value, if only even respectively uneven points in time are
considered. Below we will consider not only an asymmetric starting
distribution, but also a grid with at least one even dimension. At
least half of the $dw$ states have a positive probability on time
step $m$, if at least one corner state provides a positive probability.
Remember that this is only true for times $m$, $m+2$, $m+4$ and
so on in the event of uneven width and depth, but true for all times
$m$, $m+1$, $m+2$ and so on in the event of at least one even dimension.

According to this, the average state probability of those states with
a positive probability is not exceeding $\frac{2}{dw-1}$. It should
be clear that the corner states provide a probability less than this
in the long run because of the secluded location. So, for a particular
corner state $z^{*}$ the probability should be $P(X_{m+2}=z^{*})\leq\frac{2}{dw-1}$.
In grids of not less than $8$ states the probability $P(X_{m+2}=z^{*})$
becomes less than $\frac{2}{8-1}=\frac{2}{7}\leq\frac{1}{3}$ and
therefore independence becomes impossible. The two-step probability
remains at $\frac{1}{3}$, but in big exploration areas each state
is visited infrequently and the probability for being in a corner
state drops below $\frac{1}{3}$.

\emph{Case 3:} Boundless Random Walk

A boundless random walk enables the UAVs to leave the grid on one
side, in order to reappear on the other side. We will now examine
a grid with at least one uneven dimension. A drone can fly in one
direction and straight back in two time units. In a bordered random
walk the recurrence time to any state is always even, in contrast
in a boundless random walk with at least one uneven dimension a drone
can back to the starting point in an uneven number of time steps,
too. This reduces to the first case, by what the events are dependent.
\end{proof}
The objective is to calculate the expected coverage and therefore
it is not in someone's interest to have independent events $\left\{ X_{0}=z\right\} ,\left\{ X_{1}=z\right\} ,\left\{ X_{2}=z\right\} ,\ldots,$
but to have independent events $\left\{ X_{0}\neq z\right\} ,\left\{ X_{1}\neq z\right\} ,\left\{ X_{2}\neq z\right\} $
and so on.
\begin{claim}
\label{claim:The-events-}The events $A$ and $B$ are independent,
if and only if the complementary events $A^{c}$ and $B^{c}$ are
independent.\end{claim}
\begin{proof}
As argued before, the events $A$ and $B$ are independent, if and
only if $P(A\cap B)=P(A)\cdot P(B)$ is valid.

\begin{eqnarray*}
P(A^{c}\cap B^{c}) & = & P((A\cup B)^{c})=1-P(A\cup B)=1-(P(A)+P(B)-P(A\cap B))\\
P(A^{c})\cdot P(B^{c}) & = & (1-P(A))(1-P(B))=1-P(A)-P(B)+P(A)P(B)\\
 & = & 1-(P(A)+P(B)-P(A)P(B))
\end{eqnarray*}
The equation $P(A^{c}\cap B^{c})=P(A^{c})\cdot P(B^{c})$ is true,
if and only if $P(A\cap B)=P(A)\cdot P(B)$ is valid.\end{proof}
\begin{thm}
\label{thm:Dependence-of-Coverage}Dependence of Coverage Events (2)
\end{thm}
The events $\{X_{0}\neq z\},...,\{X_{n}\neq z\}$ are stochastically
dependent for $n\geq3$ and at least one state $z$.
\begin{proof}
The proof follows directly from proposition \ref{prop:Dependence-of-Coverage}
and claim \ref{claim:The-events-}.
\end{proof}
It was shown that the considered events $\{X_{0}\neq z\},...,\{X_{n}\neq z\}$
are stochastically dependent generally. Recall the definition
\[
C_{n,z}=\begin{cases}
1 & if\,X_{0}=z\vee X_{1}=z\vee...\vee X_{n}=z\\
0 & else
\end{cases}
\]
for the coverage of state $z$ up to time step $n$. If one takes
no notice of the dependence characteristic and assumes independence,
the coverage probability of a single state is easy to calculate. The
just made assumption leads to the coverage probability
\begin{equation}
P\left(C_{n,z}=1\right)=1-\prod_{t=0}^{n}\left(1-P(X_{t}=z\right),\label{eq:5}
\end{equation}
like written in \cite{key-1}. However, this formula cannot be true
in general, because it disregards the presence of theorem \ref{thm:Dependence-of-Coverage}.
In order to take the dependence characteristic into account, one must
specify appropriate Markov chains. Let the first arrival time in state
$z$ be $A_{z}$ with the relation
\[
A_{z}:=n\,\Longleftrightarrow\,X_{0}\neq z,X_{1}\neq z,...,X_{n-1}\neq z,X_{n}=z.
\]
State $z$ is visited until time $n$, if and only if the first arrival
in this state is at the latest on point $n$ of time. By reason of
\begin{equation}
P(C_{n,z}=1)=P(A_{z}\leq n)
\end{equation}
the coverage probabilities can be reduced to the calculation of arrival
probabilities. In order to receive the sought-after probabilities,
we need to create a new Markov chain for every state in the considered
random walk model. Subject to state $r\in E$ the associated Markov
chain is named $M_{r}$. The just introduced new Markov chains $M_{1},\ldots M_{|E|}$
accord with the Markov chain $M$ (for example as in table \ref{tab:Transition-matrix-of})
nearly complete, but there are minor differences. State space and
starting distribution coincide perfectly, but there is a change in
row $r$ of the transition matrix $M_{r},r=1,\ldots|E|$. The transition
matrix $P_{r}$ contains a one in row and column $r$ and zeros elsewhere
in row $r$, as is evident in table \ref{tab:Transition-matrix-}.

\begin{table}
$\begin{pmatrix} & \frac{1}{2} &  & \frac{1}{2}\\
\frac{1}{3} &  & \frac{1}{3} &  & \frac{1}{3}\\
 & \frac{1}{2} &  &  &  & \frac{1}{2}\\
\frac{1}{3} &  &  &  & \frac{1}{3} &  & \frac{1}{3}\\
 &  &  &  & 1\\
 &  & \frac{1}{3} &  & \frac{1}{3} &  &  &  & \frac{1}{3}\\
 &  &  & \frac{1}{2} &  &  &  & \frac{1}{2}\\
 &  &  &  & \frac{1}{3} &  & \frac{1}{3} &  & \frac{1}{3}\\
 &  &  &  &  & \frac{1}{2} &  & \frac{1}{2}
\end{pmatrix}$

\medskip{}

\caption{Transition matrix $P_{5}$ of the bordered symmetric random walk in
a $3\times3$ - grid with absorbing state $5$.\label{tab:Transition-matrix-}}
\end{table}

This change effects an altered behaviour of the Markov chain $M_{r}$
compared to $M$. State $r$ is called absorbing, because the stochastic
process remains in this state as soon as it has visited this state
for the first time. Until this event, the behaviour of Markov chain
$M_{r}$ is identical to that of the original Markov chain $M$. In
order to calculate the probability $P(A_{r}\leq n)$ in the original
Markov chain $M$, we can go over to the new Markov chain $M_{r}$
which will be used henceforward. The first arrival time does not change,
because both Markov chains behave exactly identical till the first
arrival. In the new Markov model $M_{r}$ it meets to check the location
of the random walk on point $n$ of time, because 
\begin{equation}
A_{r}\leq n\Longleftrightarrow X_{n}=r
\end{equation}
holds. The probability of the latter term can be obtained, as customary
in Markov chains, by multiplication of the starting distribution with
the exponentiated transition matrix, thus
\begin{equation}
P(C_{n,r}=1)=(\pi_{0}P_{r}^{n})_{r}\label{eq:4}
\end{equation}
and the $r$-th value of this row vector is selected. This can be
done for all states $r\in E$, but remember that every state $r$
requires its own transition matrix $P_{r}$. The expected coverage
of the random walk after $n$ time steps is the expected value 
\begin{equation}
E\left(\frac{\sum_{r\in E}C_{n,r}}{|E|}\right)=\frac{\sum_{r\in E}E\left(C_{n,r}\right)}{|E|}=\frac{\sum_{r\in E}P\left(C_{n,r}=1\right)}{|E|}=\frac{\sum_{r\in E}\left(\pi_{0}P_{r}^{n}\right)_{r}}{|E|}
\end{equation}
divided by the number of states in the Markov model. It is about the
ratio of the states in the Markov chain that has been visited after
$n$ time steps, that is to be expected. Altogether the handling with
dependent events leads to a much more complicated solution of the
expected coverage as though one ignores this pre-assumption.

\subsection{Independent Mobility}

So far, the expected coverage of a single UAV was calculated. However,
one of the greatest advantages of the operation of UAVs in contrary
to helicopters is that the former can be used in groups, even at places
with little room. Together they can explore the target area faster,
especially if they cooperate. For the sake of convenience we limit
on independent mobility, so that each drone uses the same Markov model
$M$ for exploring, independent of the movement of the other UAVs.
Thereby the coverage probability in \ref{eq:4} changes to
\begin{equation}
P\left(C_{n,r}^{multi}=1\right)=1-\left(1-P\left(C_{n,r}=1\right)\right)^{|E|}=1-\left(1-\left(\pi_{0}P_{r}^{n}\right)_{r}\right)^{|E|}.
\end{equation}

\subsection{Comparison of Analytic and Simulation Results}

The analytic solution of the expected coverage was deduced on the
whole. The analytic results are confirmed by MCMC simulations which
can be easily coded once the transition probabilities are available.
However, there is a great gap between the analytic formula \ref{eq:5},
that ignores the dependence of events, and the correct analytic results
\ref{eq:4} as well as the simulation results. In \cite{key-1} it
is concluded that this peculiarity is due to error propagation. This
argument can be excluded, because the correct formula \ref{eq:4}
and the simulation results show no deviation and therefore are in
agreement. On grounds of numerical computability, MCMC- simulations
should be preferred opposite to the analytic formula, because the
simulations can provide approximately exact results and can handle
Markov models with huge state spaces considerably better.

\section{Conclusions}

In this work, we introduced Markov chains with and without borders
in two and three dimensions. We deduced an analytic solution of the
expected coverage in case of independent mobility. Though it was proven
that certain events are not independent, which complicates the calculation
vastly. The non-consideration of this fact leads to an incorrect solution
for the expected coverage which cannot be explained by numerical inaccuracy. Hence, MCMC simulations provide still a good approximation for the
analytic results.

\section*{Appendix: Sequencing in Three-Dimensional Grids}

The sequencing and computation of the transition matrix for the bordered
symmetric random walk in three-dimensional grids are programmed in
R \cite{key-3}.

 \begin{lstlisting}
# The states are sequenced firstly on positive x-direction,
# then in positive y-direction and eventually in positive z-direction.
# That means, the node with the coordinates (x,y,z)
# holds the number x+w(y-1)+wd(z-1).

calculateTransitionMatrix3d = function ( width , depth , height ) {
  w = width
  d = depth
  h = height
  numberOfStates = w*d*h
  P = matrix ( 0 , nrow = numberOfStates , ncol = numberOfStates )
  
  nodePoints = matrix ( 0 , nrow = numberOfStates , ncol = 3 )
  nodeNumber = 0
  for ( z in 1:h ) {
    for ( y in 1:d ) {
      for ( x in 1:w ) {
        nodeNumber = nodeNumber + 1
        nodePoints [ nodeNumber , ] = c (x,y,z)
      }
    }
  }
  
  calculateNeighbourNodes = function ( x , y , z ) {
    neighbourNodes = matrix ( nrow = 0 , ncol = 3)
    if ( x > 1 )
      neighbourNodes = rbind ( neighbourNodes , c ( x-1 , y , z ) )
    if ( x < w )
      neighbourNodes = rbind ( neighbourNodes , c ( x+1 , y , z ) )
    if ( y > 1 )
      neighbourNodes = rbind ( neighbourNodes , c ( x , y-1 , z ) )
    if ( y < d )
      neighbourNodes = rbind ( neighbourNodes , c ( x , y+1 , z ) )
    if ( z > 1 )
      neighbourNodes = rbind ( neighbourNodes , c ( x , y , z-1 ) )
    if ( z < h )
      neighbourNodes = rbind ( neighbourNodes , c ( x , y , z+1 ) )
    
    return ( neighbourNodes )
  }
  
  neighbourNodes = list ()
  nodeDegrees = c()
  nodeNumber = 0
  for ( z in 1:h ) {
    for ( y in 1:d ) {
      for ( x in 1:w ) {
        nodeNumber = nodeNumber + 1
        neighbourNodes [[ nodeNumber ]]
			= calculateNeighbourNodes ( x , y , z )
        nodeDegrees [[ nodeNumber ]]
			= nrow (neighbourNodes [[ nodeNumber ]] )
      }
    }
  }
  
  nodeTransformation = function ( x , y , z ) {
    return ( x + w * (y-1) + w * d * (z-1) )
  }
  
  transformedNeighbourNodes = list()
  nodeNumber = 0
  for ( z in 1:h ) {
    for ( y in 1:d ) {
      for ( x in 1:w ) {
        res = c()
        nodeNumber = nodeNumber + 1
        for ( i in 1:nrow ( neighbourNodes [[ nodeNumber ]] ) ) {
          value1 = neighbourNodes [[ nodeNumber ]] [i,1]
          value2 = neighbourNodes [[ nodeNumber ]] [i,2]
          value3 = neighbourNodes [[ nodeNumber ]] [i,3]
          res = c ( res , nodeTransformation
			( value1 , value2 , value3 ) )
        }
        transformedNeighbourNodes [[ nodeNumber ]] = res
      }
    }
  }
  
  for ( nodeNumber in 1:numberOfStates ) {
    P [ nodeNumber , transformedNeighbourNodes [[ nodeNumber ]] ]
		= 1 / nodeDegrees [ nodeNumber ]
  }
  
  return ( P )
}
\end{lstlisting}
\end{document}